\documentclass[10pt,twocolumn,twoside]{IEEEtran}
\usepackage[left=0.75in,right=0.75in,top=0.73in,bottom=.76in]{geometry}
\usepackage[utf8]{inputenc}
\usepackage[T1]{fontenc}
\usepackage{url}
\usepackage{ifthen}
\usepackage{multicol}
\usepackage{cite}

\usepackage{graphicx,epstopdf,amssymb,amsthm}% ,epstopdf,spconf,epsfig,fullpage}%, ,graphicx,psfrag,url,amsmath,amsthm,comment}
\usepackage{ulem,color}
\usepackage{epstopdf}
\usepackage[utf8]{inputenc}
\usepackage[english]{babel}
\usepackage{caption}
\usepackage{algorithm}
\usepackage{algpseudocode}
\usepackage{amsmath,amsfonts,amssymb,amsthm}
\usepackage{booktabs}
\usepackage{lipsum}
\usepackage{quiver}
\usepackage{svg}
\newtheorem{example}{Example}

\newtheorem{theorem}{Theorem}
\newtheorem{problem}{Problem}
\newtheorem{lemma}{Lemma}

\newtheorem{corollary}{Corollary}

\newtheorem{remark}{Remark}

\newtheorem{definition}{Definition}
\usepackage{cite}
\usepackage[T1]{fontenc}
\usepackage{color}
\usepackage{amsmath}
\usepackage{amsthm}
\usepackage{amssymb}
\usepackage{stmaryrd}
\usepackage{graphicx}
\usepackage{esint}

\makeatletter
\begin{document}

%Here goes the title

\title{\vspace{0.25in}Characterizing Compositionality of LQR from \\ the Categorical Perspective}

\author{%\IEEEauthorblockN{
Baike She, %\IEEEauthorrefmark{2}, 
Tyler Hanks, James Fairbanks, %\IEEEauthorrefmark{1},
and %\IEEEauthorrefmark{2},
Matthew Hale*%\IEEEauthorrefmark{2}
\thanks{*Baike She and Matthew Hale are with the Department of Mechanical and Aerospace Engineering at the University of Florida.
Their work was supported in part by AFOSR under grant
FA9550-19-1-0169; Tyler Hanks and James Fairbanks are with the
Department of Computer and Information Science and Engineering at University of Florida. Baike She, Matthew Hale, and James Fairbanks, were supported by DARPA under award no. HR00112220038. In addition,  Tyler Hanks was supported by the National Science Foundation Graduate
Research Fellowship Program under Grant No. DGE-1842473. Any opinions, findings,
and conclusions or recommendations expressed in this material are those of the author(s)
and do not necessarily reflect the views of the NSF. 
E-mails: \texttt{\{shebaike,t.hanks,fairbanksj,matthewhale\}@ufl.edu}.
}
%Eldon Tyrell\IEEEauthorrefmark{4},~\IEEEmembership{Fellow,~IEEE}}

%\IEEEauthorblockA{\IEEEauthorrefmark{1}
%  Department of Electrical and Computer Engineering, Purdue University, West Lafayette, IN, 47907 USA
% \IEEEauthorblockA{\IEEEauthorrefmark{2} School of Electrical and Computer Engineering, Purdue University, West Lafayette, IN, 47907, USA}
%\IEEEauthorblockA{\IEEEauthorrefmark{3}Starfleet Academy, San Francisco, CA 96678 USA}
%\IEEEauthorblockA{\IEEEauthorrefmark{4}Tyrell Inc., 123 Replicant Street, Los Angeles, CA 90210 USA}% <-this % stops an unwanted space
%\thanks{Manuscript received December 1, 2012; revised August 26, 2015. 
%Corresponding author: M. Shell (email: http://www.michaelshell.org/contact.html).}}
}

%\author
%{\IEEEauthorblockN{Author 1}
%\IEEEauthorblockA{School of Electrical and\\Computer Engineering\\
%University\\
%Location\\
%Email: }
%\and
%\IEEEauthorblockN{Author 2}
%\IEEEauthorblockA{University\\
%Location\\
%Email: }
%}
\maketitle

%Main body starts

\begin{abstract}
Composing systems is a fundamental concept in modern control systems, yet it remains challenging to formally analyze how controllers designed for individual subsystems 
can differ from controllers designed for the composition of those subsystems.
%will be affected after the systems are composed, compared to designing controllers directly for the composed system. 
To address this challenge, we propose a novel approach to composing control systems based on \textit{resource sharing machines}, a concept from applied category theory. 
%We study composite systems that are formed by the composition of two or more subsystems. 
We use resource sharing machines to investigate the differences between 
(i) the linear-quadratic regulator (LQR) designed directly for a composite system and 
(ii) the LQR that is attained through the composition of LQRs designed for each subsystem. 
%(i) the composition of LQR designs Linear-Quadratic Regulator (LQR) for the composite system that is attained through composition of the LQR designs of its subsystems 
%and (ii) the LQR designed directly for the composite system. 
We first establish novel formalisms to
compose LQR control 
designs using resource sharing machines. Then we develop new sufficient conditions to guarantee that 
the LQR designed for a composite system is equal to 
the LQR attained through composition of LQRs
for its subsystems. In addition, we reduce the developed condition to that of checking the controllability
and observability of a certain linear, time-invariant system, which provides a simple, computationally
efficient procedure for evaluating the equivalence of controllers for composed systems.

\end{abstract}

%\begin {IEEEkeywords}
%IoT, Ontology, Semantics,  SSN, OWL, OBOE, OpenIoT, SWEET, SUMO
%\end{IEEEkeywords}
\section{Introduction}
\label{intro}
Modern technology has made it simple to design and construct large-scale dynamical systems through the composition of many subsystems. For example, smart communities may combine control systems such as networks of smart grids, water distribution systems, and transportation systems to provide rich, complex capabilities. 
It is therefore crucial to understand how to perform control design for subsystems and compose these controllers for the composite system\footnote{Throughout this paper, we use the term \textit{composite system} to refer to a system that is formed through the composition of many subsystems.}. 
%We use the term \textit{composite system} to refer to a system that is formed through the composition of many subsystems. 

In the control community, it has been observed that many systems can be analyzed by analyzing their subsystems \cite{li2011smart}, where some properties of the subsystems are inherited by the 
composite system. 
%The ability to infer properties of the composed system from
%its subsystems is related to the notion of \textit{compositionality} in applied category theory \cite{fong2019invitation}, which is broadly concerned
%with the retention of properties of mathematical objects under various domain-specific forms of composition. 
%Although the category-theoretic
%concept of compositionality has not been extensively explored in the area of controls, previous researchers from the control community have 
%analyzed similar ideas
%from alternative perspectives. 
For example, the stability of composite feedback systems can be analyzed based on their subsystems using the small-gain theorem and composed Lyapunov functions \cite{haddad2011nonlinear}, and parallel interconnections of passive systems will induce a composite passive system \cite{ebenbauer2009dissipation}. Additionally, \cite{chapman2014controllability,tran2018generalized, she2020characterizing} have explored how controllability of networks can be explored through their sub-networks. However, the lack of formal analytical language in composing control systems makes it challenging to further advance and unify this work.

Recent advances in hybrid systems have used \textit{applied category theory} \cite{fong2019invitation}  as a formal language for describing the composition of 
continuous- and discrete-time systems~\cite{culbertson2019formal, bakirtzis2021categorical,kvalheim2021conley, lerman2016category}. 
Applied category theory is also being used in mathematical modeling \cite{libkind2022algebraic}, 
scientific computing \cite{halter2020compositional}, data science \cite{patterson2021categorical}, and more. 
Roughly speaking, 
applied category theory focuses on the abstract connections between objects, rather than studying the objects themselves  \cite{awodey2010category,vagner2014algebras}.
Although category theory has found success in
various areas of engineering, there has been relatively limited work on using it to study
compositions of control systems; a representative sample of existing works includes~\cite{tabuada2005quotients, lerman2018networks, lerman2020networks}. 

% The recent control design of hybrid systems has involved the usage of \textit{applied category theory} \cite{fong2019invitation} as a formal language to describe the composition the systems from both the continuous and discrete sides \cite{culbertson2019formal, bakirtzis2021categorical,kvalheim2021conley, lerman2016category}. Besides leveraging applied category theory in hybrid systems, it is also a trend to implement applied category theory in the area of mathematical modeling \cite{libkind2022algebraic}, scientific computing \cite{halter2020compositional}, and data science \cite{patterson2021categorical}, etc. 
% % can add one or two examples
% Applied category theory mainly focuses on the abstract connections between objectives rather than studying the objectives themselves \cite{awodey2010category}. Hence, Applied category theory can offer formal analytical language and tools to explore the compositionality of dynamical systems \cite{vagner2014algebras}. 
% how can we transit to LQR problems more natually? % we might need more reference to illustrate the importance of category theory.

We focus on exploring a specific way of composing dynamical systems known as \textit{resource sharing machines} (RSMs) \cite{libkind2021operadic, libkind2020algebra}, which are defined in the language of applied category theory. RSMs present a way of composing dynamical systems in which, in a precise way, the states of the composite system can preserve the states of the subsystems, and vice versa. 
%As a result, composing subsystems through the lens of
Thus, the use of
RSMs can enable new analyses for a composite system through combining analyses originally designed for its subsystems~\cite{libkind2021operadic, libkind2020algebra}. 

In this paper, we study the relationship between (i) the composition of the control designs of 
a collection of subsystems and (ii) control designs made directly for the composition of those subsystems. In both cases, the composition of subsystems is formalized by RSMs. In particular, we will focus on the linear-quadratic regulator (LQR) because it is a cornerstone of control theory~\cite{anderson2007optimal} and the focus of ongoing research~\cite{dorfler2022role, zhang2023revisiting}. 
Additionally, LQR has a closed-form solution, which lets us
study the \textit{compositionality} of the LQR, which we define as the 
property that the composition of LQR controllers designed for subsystems
is equivalent to the LQR controller designed for the composition of those subsystems\footnote{A formal mathematical definition of ``compositionality'' in this context is given in category-theoretic terms in  Section~\ref{sec:background}.}.

To summarize, our contributions are: 
 \begin{itemize}
 \item
We present a novel approach to composing control systems using resource sharing machines (RSMs). Specifically, we compose LQR control designs using RSMs
   % \item We introduce a new way of composing dynamical systems through the idea of resource sharing machines. We propose a way of composing control designs (LQR) through the resource sharing machine;
   \item We show that LQR is \textit{not} always compositional.
    \item We introduce a new, generalized Riccati equation, and we use it to derive new sufficient conditions under which the compositionality of LQR is guaranteed. 
    % \item We show that although the resource sharing machine induces the compositionality of the  subsystems without control, 
    % it is not necessary that the compositionality will be preserved under the LQR design. We further developed novel conditions to characterize the  compositionality of LQR under the setting of resource sharing machines. 
    %\item{In addition to presenting our results in the language of control theory, we will also express them in the framework of applied category theory.}
    % \item \baike{Apart from developing the results in the language of control theory, we will also formulate the results in the language of applied category theory.} %benefit?
    \item We reduce the problem of validating compositionality of LQR designs to that of checking the controllability and observability
    of a certain linear, time-invariant system. 
\end{itemize}

The rest of the paper is structured as follows. Section~II provides background and problem statements. Section III provides results on composing linear control systems through resource sharing machines (RSMs). Section~IV examines the compositionality of LQR via RSMs. Section~V concludes.

% The rest of the paper is organized as follows. We introduce the background and problem statements in Section~II. 
% In Section~III, we propose and compose
% linear control systems through
% resource sharing machines. 
% In Section~IV, we explore the compositionality of LQR via resource sharing machines. %\baike{In Section~V, we leverage the language from applied category theory to formalize the results developed in this work.}  %future works
% Then Section~V summarizes the work and gives potentially future research directions. 
\paragraph*{Notation}
We use~$\mathbb{R}$ and~$\mathbb{N}$ to denote the real and natural numbers, respectively.
All vectors are columns. 
We use $\delta_i$ to denote a basis vector with $i^{th}$ entry $1$ and all others~$0$.  
We define $\underline{n} := \{1,2,3, \dots, n\}$. 
%For two sets  $\underline{n_1}$ and $\underline{n_2}$, we use the difference set $\underline{n_1+n_2}\backslash\underline{n_1}$ to denote $\{n_1+1, n_1+2,\dots, n_1+n_2\}$. 
For a matrix $A$, we use $A_{:,j}$ and 
$A_{i,:}$ to denote its $j^{th}$ column and $i^{th}$ row, respectively. We use $x(t)$ and $x$ interchangeably to denote a state vector at time~$t$. We use 
the terms ``variable'' and ``state'' interchangeably. We use $\boldsymbol{0}$ to denote a vector or matrix of zeroes, and its dimension will be clear from context.
%\baike{Put more notations here.}

\section{Preliminaries and Problem Formulation}\label{section2}
In this section, we first provide background from resource sharing machines and the composition of linear systems. 
Then we state the problems that are the focus of this paper. 
% This section provides background on applied category theory, resource sharing machines, and composition of linear systems.  %the network of reproduction numbers. 
% %We will include several existing results to lay the foundation for defining distributed reproduction numbers and using them to analyze spreading behaviors later in the paper. 
% Then we formulate the problems that are the focus of this work. 

\subsection{Applied Category Theory and Resource Sharing Machines} \label{sec:background}

%\tyler{I'm trying to go light/null on the actual category theory and instead emphasize the intuition behind the graphical syntax considering our audience and space restrictions. Organizationally speaking, I think this means we can ship alot of later definitions up here (e.g. definition of composite open loop systems, definition of composite closed loop systems) and then formally define what it means for LQR to be compositional (i.e. when it satisfies that informal commutative diagram I drew). Thoughts? Also keep in mind that as requested, the prose is mostly word vomit for now just to get it all down!}

In this paper, we focus on the compositional properties of LQR control design for coupled dynamical systems. In this section, we present a graphical calculus for specifying such coupled systems based on ``resource sharing'' as defined in \cite{libkind2021operadic}. We extend it to closed-loop control designs and define what it means for an LQR controller to be \textit{compositional}.

\begin{definition}
    Let~$X=\mathbb{R}^n$ and $U=\mathbb{R}^m$ for some~$m, n \in \mathbb{N}$. The symbol~$TX\cong X\times X$ denotes the tangent bundle of $X$. An \textbf{open-loop control system} is a $U$-parameterized vector field on $X$, i.e., a smooth function $f\colon U\times X\rightarrow TX$, where $f(u,x) = (x,d)$ for all $x\in X$ and $u\in U$ \cite{spivak2015dynam}. Here $d\in X$ is the tangent vector at the point $x$. The set
    $X$ is called the \textbf{state space} of the system and $U$ is called the \textbf{control surface}. Such a vector field 
    defines a system of ordinary differential equations
    %\begin{equation*}
        $\dot{x} = f(x,u)$,
    %\end{equation*}
    which we call the \textbf{dynamics} of the system. In the case of a linear open-loop control system, the ODE system has the form
    %\begin{equation*}
        ${\dot{x} = Ax + Bu}$
    %\end{equation*}
    for some $A\in \mathbb{R}^{n\times n}$ and $B\in \mathbb{R}^{n\times m}$.
\end{definition}

We use the term \textit{subsystem} to refer to a system that is not a composition of other systems. 
To model coupled dynamical systems, each open-loop subsystem must specify its \textit{boundary}, a subspace of its state space designating which variables can couple to other systems. The intuition is that given some open-loop control subsystems with boundary variables, we can ``glue" the subsystems together along shared boundary variables to form a composite system. We can represent an open-loop control subsystem with a boundary graphically. For example, consider a system~$S_1$ with state space $\mathbb{R}^3$, control surface $\mathbb{R}$, and dynamics $(\dot{x},\dot{y},\dot{z}) = f(u,(x,y,z))$, where $x$ and $z$ are boundary variables. This is represented by the following
\textit{undirected wiring diagram} (UWD): 

%\begin{figure}
%\centering
\begin{center}
\includegraphics[scale=0.18]{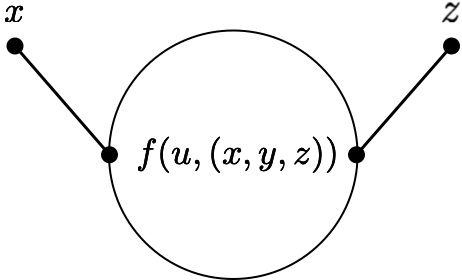}.
\end{center}
%\caption{An undirected wiring diagram (UWD) representing the system~$S_1$ with three states and two boundary variables.}
%\label{fig:atomicuwd}
%\end{figure}
The circle represents the open-loop subsystem defined by $f$, and the wires represent its boundary variables,~$x$ and~$z$. 
\subsection{Composing Open-Loop Control Subsystems}
To compose open-loop control subsystems diagrammatically, 
we simply connect the wires for the boundary variables that are coupled. For example, suppose we want to compose $S_1$ with another subsystem $S_2$ that has state space $\mathbb{R}^2$, control surface $\mathbb{R}^2$, and dynamics $(\dot{w},\dot{z}) = g(u_2,(w,z))$ with both of its state variables in its boundary. Suppose that we want to couple the third state of $S_1$  (i.e., $z$ in $f$) with the second state of $S_2$  (i.e., $z$ in $g$), we have the resultant UWD shown below: %in Figure~\ref{fig:compositeUWD}. 

%\begin{figure}
%\centering
\begin{center}
\includegraphics[scale=0.18]{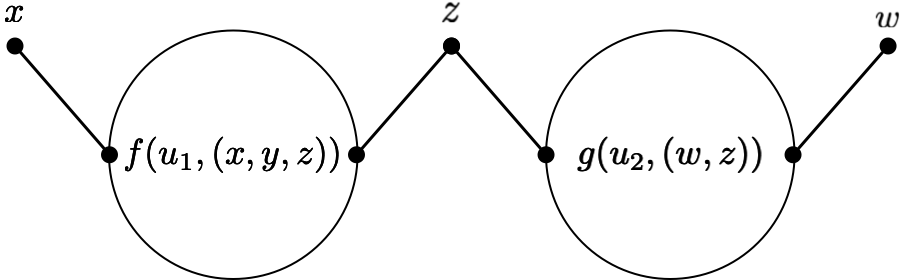}.
\end{center}
%\caption{An undirected wiring diagram (UWD) representing the composition of the systems~$S_1$ and~$S_2$.}
%\label{fig:compositeUWD}
%\end{figure}
Note that we consider compositions of two subsystems for simplicity; 
our results generalize to any number of subsystems in an obvious way.

A benefit of this graphical syntax is that it uniquely defines the dynamics of the composite system based on the dynamics of the subsystems. For linear systems, the composition is achieved by the construction of a \textit{composition matrix} $K$. To define it, consider two open-loop linear control subsystems $S_1$ and $S_2$, where
%\begin{itemize}
    %\item 
    $S_1$ has state space $\mathbb{R}^{n_1}$, control surface $\mathbb{R}^{m_1}$, and dynamics 
    \begin{equation} \label{Eq:Sys_1}
        \dot{x}^1 (t) = {A}^1 x^1(t) + {B}^1u^1(t),
    \end{equation}
    %\item 
    and $S_2$ has state space $\mathbb{R}^{n_2}$, control surface $\mathbb{R}^{m_2}$, and dynamics 
    \begin{equation}\label{Eq:Sys_2}
        \dot{x}^2 (t) = {A}^2 x^2(t) + {B}^2u^2(t).
    \end{equation}
%\end{itemize}
Note that $x^i\in\mathbb{R}^{n_i}$, $A^i\in\mathbb{R}^{n_i\times n_i}$, $B^i\in\mathbb{R}^{n_i\times m_i}$, $u^i\in\mathbb{R}^{m_i}$, $i\in\underline{2}$.
Suppose further that $S_1$ and $S_2$ share $k \leq \min\{n_1,n_2\}$ states defined by the composition pattern of some UWD. 
This sharing allows us to define a lower-dimensional state vector for the composite system, denoted~$\bar{x}$. That is,
the composition pattern between~$S_1$ and~$S_2$
defines an identification between the new state variable $\bar{x}\in \mathbb{R}^{n_1+n_2-k}$ and the state variables of~$S_1$ and~$S_2$ themselves. 
The composition pattern 
%defines whether $\bar{x}_i = x^1_j$ for some~$j\in\emph{n_1}$
%or $\bar{x}_i = x^2_k$ for some~$k\in\emph{n_2}$ or both. 
%This composition pattern can be used to construct $K \in \mathbb{R}^{(n_1+n_2)\times(n_1+n_2-k)}$ as follows. 
%%\tyler{Baike, I tried to do some reworking of this section to make it more clear. I think it's correct but you should double check.}
%For each $i\in \underline{(n_1+n_2-k)}$,
can be used to construct the composition matrix~${K \in \mathbb{R}^{(n_1+n_2)\times(n_1+n_2-k)}}$ as follows: 
\begin{itemize}
    \item if $\bar{x}_i$ is a non-shared variable and $\bar{x}_i=x^1_j$ for some $j\in \emph{n_1}$ then we have $K_{j,:} =\delta_i^{\top}$,
    \item if $\bar{x}_i$ is a non-shared variable and $\bar{x}_i=x^2_k$ for some $k\in \emph{n_2}$ then we have $K_{k+n_1,:} =\delta_i^{\top}$,
    \item if $\bar{x}_i$ is a shared variable with $\bar{x}_i=x^1_j=x^2_k$ for some $j\in\emph{n_1}$ and~$k\in\emph{n_2}$, then we have $K_{j,:} = K_{k+n_1,:} =\delta_i^{\top}$. 
\end{itemize}

%Now that we have shown how to construct the composition matrix, 
We can now define the composition of open-loop linear subsystems based on the UWD syntax.

\begin{definition}[Composition of linear open-loop control subsystems]
\label{Def:op-contr}
    Let $S_1$ and $S_2$ be the linear open-loop control subsystems in~\eqref{Eq:Sys_1} and~\eqref{Eq:Sys_2}, respectively. Suppose we have computed the composition matrix $K$ based on a UWD defining $S_1$ and $S_2$'s composition pattern. The composite system $S_1\square S_2$ %\tyler{not sure what notation to use for this} 
    has state space $\mathbb{R}^{n_1+n_2-k}$, control surface $\mathbb{R}^{m_1+m_2}$, and dynamics
\begin{align}\label{Eq:Contr_comp}
 \dot{\bar{x}}(t)  = 
 \underbrace{
 K^{\top}
\underbrace{
\begin{bmatrix}
A^1 & \boldsymbol{0} \\
\boldsymbol{0}& A^2 \\
\end{bmatrix}}_{\bar{A}}K}_{\mathcal{A}}
\bar{x}(t) + 
\underbrace{K^{\top} \underbrace{
\begin{bmatrix}
B^1 & \boldsymbol{0} \\
\boldsymbol{0} & B^2 \\
\end{bmatrix}}_{\bar{B}}}_\mathcal{B}
\underbrace{
\begin{bmatrix}
u^1(t) \\
u^2(t) \\
\end{bmatrix}}_{u(t)}.
\end{align}
%We emphasize that~$\dot{\bar{x}}$ does not have separable dynamics
%because~$K^{\top}\bar{A}K$ and~$K^{\top}\bar{B}K$ are non-diagonal.
\end{definition}

\begin{remark}
We point out two aspects of~\eqref{Eq:Contr_comp}. First,
the dynamics are not separable because the system matrix of~\eqref{Eq:Contr_comp} is $K^\top\bar{A}K$, 
which is non-diagonal. 
Second, 
%Unlike sharing some states when composing subsystems $A^1$ and $A^2$, for the composed open-loop system in \eqref{Eq:Contr_comp}, 
control inputs are not shared. 
    Instead, the control input of a shared state variable is obtained by summing up the control inputs from the subsystems
    that share that state. 
    %the control input of a shared variable will be the summation of the corresponding control inputs from the subsystems.
\end{remark}

\begin{remark}
    The composition matrix $K$ encodes the resource sharing machine between two linear subsystems. In \eqref{Eq:Contr_comp}, the term $K\bar{x}(t)$ replicates the shared variables while preserving the non-shared variables, and it maps the replicated shared variables to the dynamics of
    each of the two subsystems. 
    %Therefore, $K\Bar{x}(t) = [x_1^\top (t) \ \ x_2^\top(t)]^\top$. 
    %Note that the equations captured by $\Bar{A}K\Bar{x}(t)$ are equivalent to the dynamics of 
    %subsystems \ref{Eq:Sys_1} and \ref{Eq:Sys_2}. Mathematically, we have the relation $[\dot{x}_1^\top (t) \ \ \dot{x}_2^\top(t)]^\top = \Bar{A}K\Bar{x}(t)$. Then the %transpose of the composition matrix, $K^\top$, composes the dynamics of the subsystems, giving
    %$\dot{\Bar{{x}}}(t)=K^\top \Bar{A}K\Bar{x}(t) = K^\top [\dot{x}_1^\top (t) \ \ \dot{x}_2^\top(t)]^\top$. 
    Using this formulation, the dynamics of a shared variable in the composite system
    can be obtained by summing the dynamics of the corresponding duplicated variables over each subsystem in which they appear. 
    This is done in the term~$K^{\top}\bar{A}K\bar{x}(t)$. 
    On the other hand, the dynamics of a non-shared variable 
    in the composite system are identical to those of the corresponding variable in the subsystem.
    %Based on this formulation, the dynamics of a shared variable is obtained from the summation of the corresponding dynamics of the duplicated variables. The dynamics of a non-shared variable is the same as the dynamics of the corresponding variable in the subsystem.
\end{remark}

To reason about LQR on composite systems, we need a formal setting for reasoning about feedback systems. For this, we use the notion of $(I,O)$-Systems developed in \cite{spivak2015dynam}.

\begin{definition}
    Let $I$ and $O$ be Euclidean spaces. An \textbf{$(I,O)$-System} is a 3-tuple $(X,f^{upd},f^{rdt})$, where
    \begin{itemize}
        \item $X$ is the system's state space and is a Euclidean space,
        \item $f^{upd}$ is an $I$-parameterized vector field $I\times X\rightarrow TX$ called the update function defining the system dynamics,
        \item $f^{rdt}$ is a function $X\rightarrow O$ called the readout function.
    \end{itemize}
\end{definition}
Observe that given an open-loop control system $S$ with state space $X$, control surface $U$, and dynamics $f$, we can construct a $(U,X)$-System as the~$3$-tuple $(X,f,\text{id}_X)$, where the
readout map~$\text{id}_X$ is the identity map on~$X$, which corresponds to a system's output simply equaling its state vector.
We call this construction a ``boxed open-loop system,'' 
and we represent it graphically by wrapping an open-loop control system in a box with input and output wires:
\begin{center}
    \includegraphics[scale=.2]{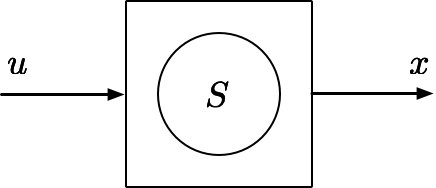}.
\end{center}
\subsection{Compositionality of LQR}
We can now define closed-loop feedback control subsystems for boxed open-loop subsystems.
\begin{definition}
\label{Def: Close}
    Given a $(U,X)$-System $(X,f\colon X\times U\rightarrow TX,\text{id}_X)$ and a full state feedback controller $F\colon X\rightarrow U$, the \textbf{closed-loop feedback control system} defined by this data is a $(\emptyset,X)$-System with state space~$X$, 
    update function given by $x \mapsto f(F(x),x)$, and readout function~$\textnormal{id}_X$.
    %\begin{itemize}
    %    \item state space $X$,
    %    \item update function given by $x \mapsto f(F(x),x)$,
    %    \item readout function that is the identity on $X$.
    %\end{itemize}
        Graphically, we represent the system by
    \begin{center}
        \includegraphics[scale=.2]{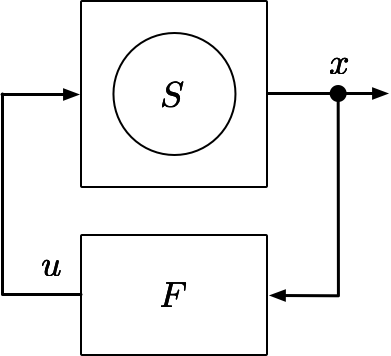}.
    \end{center}
    In the linear case, this gives~$\dot{x} = Ax + BFx$,
    where~$F\in \mathbb{R}^{m\times n}$ is a gain matrix.
    Note that the input space is the empty set to signify the closed-loop nature of the system, i.e.,
    each input is uniquely determined by the system's state.
\end{definition}

%\tyler{TODO: not sure if we want to move the definition of composing closed-loop control systems up here?}

In light of the above definitions, we now have two ways of designing LQR controllers for composite dynamical systems: we can either design controllers for each subsystem individually, then compose the resultant controllers, or compose the subsystem dynamics, and design a single controller for the overall composite system. This observation sets up the central definition of the paper.

\begin{definition}[Compositionality of LQR]
\label{Def:compositionality}

Given open-loop subsystems $S_1$ and $S_2$, an LQR controller is \textbf{compositional} if the following diagram commutes:

% https://q.uiver.app/?q=WzAsNCxbMCwwLCJTXzEsU18yIl0sWzIsMCwiTFFSKFNfMSksTFFSKFNfMikiXSxbMCwyLCJTXzFcXGNkb3QgU18yIl0sWzIsMiwiTFFSKFNfMVxcY2RvdCBTXzIpIl0sWzAsMSwiTFFSLExRUiJdLFswLDIsIlxcdGV4dHtjb21wb3NlfSIsMl0sWzIsMywiTFFSIiwyXSxbMSwzLCJcXHRleHR7Y29tcG9zZX0iXV0=
\begin{tikzcd}
	{S_1,S_2} && {LQR(S_1),LQR(S_2)} \\
	\\
	{S_1\square S_2} && {LQR(S_1\square S_2)}
	\arrow["{LQR,LQR}", from=1-1, to=1-3]
	\arrow["{\text{compose}}"', from=1-1, to=3-1]
	\arrow["LQR"', from=3-1, to=3-3]
	\arrow["{\text{compose}}", from=1-3, to=3-3]
\end{tikzcd}
where~``$LQR$'' denotes the synthesis of an LQR controller to minimize a given cost,
and~``$compose$'' denotes the composition of open-loop systems
on the left of the diagram and composition of closed-loop systems on the right of the diagram.
\end{definition}

Unpacking this definition, we have that in order for an LQR controller to be compositional for given subsystems $S_1$ and $S_2$, the controller of the composite of the subsystems must be the same as composing the controllers designed for $S_1$ and $S_2$ individually. 
%The rest of this paper is devoted to studying under what conditions this property holds.

%\tyler{I think the categorical machinery we need to talk about controlled dynamical systems is the category of linearly parameterized dynamical systems, which is described by Evan in \cite{epatters2023}}

%\baike{We need background on resource sharing machines, especially some core theorems/definitions that define them. Before introducing resource sharing machines, we need backgrounds on tools from applied category theory that we will use in this paper. (It will be nice to write these content within a single column.)} 
%\mh{We should formally define the notion of ``compositionality'' for RSMs here. This should go under
%a formal ``Definition'' heading so that we can refer back to it later.}
%\mh{Let's also formally define what it means to ``share'' a state.}

\subsection{Problem Statements}
In light of the above, we must formalize the composition of closed-loop systems, then evaluate
the compositionality of their properties. In particular, we will solve the following: 
\begin{problem}
\label{prob:1}
Under the definition of closed-loop control systems in~\ref{Def: Close},
how can we compose closed-loop control inputs using the framework of resource sharing machines?
\end{problem}
\begin{problem}
\label{prob:2}
Under the definition of compositionality in Definition~\ref{Def:compositionality},
does compositionality of stability, which holds for systems without control, also hold when control inputs are present?
That is, must the composition of stable subsystems give a stable composite system? 
% Compared to the compositionality of systems without control, will the property of compositionality remain under the existence of control inputs?
\end{problem}
%\mh{Inside this problem, refer back to the definition of compositionality in Section II-A (once it's added there).}
\begin{problem}
\label{prob:3}
Under what conditions will the compositionality of the LQR design given by Definition~\ref{Def:compositionality} be preserved when using the resource sharing machines formulation?
%Given one popular optimal control design technique, LQR, under what conditions will the compositionality of the LQR design be preserved via the formulation of the resource sharing machines? 
\end{problem}
% \begin{problem}
% \label{prob:4}
% Based on the idea of resource sharing machines in the area of applied category theory, how can we characterize the compositionality results developed in this work from the categorical perspective?
% \end{problem}
%Maybe add this result later for a journal extension.
We will solve Problems~\ref{prob:1}-\ref{prob:3} in the next two %(\baike{three, if we have another section of leverage category theory to characterize the results}) 
sections.

\subsection{Illustrative Example}
We close this section with the following example to illustrate the composition of two linear systems with no control inputs
through the composition matrix~$K$; we omit inputs for simplicity, but composition of systems with inputs
can be done using the same composition matrix~$K$ constructed in the example (as shown in Definition~\ref{Def:op-contr}). 
\begin{example}
Consider the following two subsystems
\begin{align} 
   \underbrace{
   \begin{bmatrix}
       \dot {x}^i_{1}(t) \\ \dot{x}^i_{2}(t)
   \end{bmatrix}}_{\dot{x}^i(t)}
     &= 
    \underbrace{
   \begin{bmatrix}
       a^i_{11} & a^i_{12}\\a^i_{21} & a^i_{22}
   \end{bmatrix}}_{A^i}
   \underbrace{
    \begin{bmatrix}
       {x}^i_{1}(t) \\ {x}^i_{2}(t)
   \end{bmatrix}}_{x^i(t)},
\end{align}
where $i\in \underline{2}$. 
In this example, the two subsystems share $x^1_2$ and $x^2_1$, i.e., $x^1_2=x^2_1=\Bar{x}_2$ in the composed system. Based on \eqref{Eq:Contr_comp},
we have the composite system
\begin{align} 
   \underbrace{
   \begin{bmatrix}
       \dot {\bar{x}}_{1}(t) \\ 
       \dot {\bar{x}}_{2}(t)\\
       \dot {\bar{x}}_{3}(t)
   \end{bmatrix}}_{\dot{\bar{x}}(t)}
     &= 
     \underbrace{ 
    \begin{bmatrix}
       1 & 0 & 0 & 0\\ 
       0 & 1 & 1 & 0\\
       0 & 0 & 0 & 1\\
   \end{bmatrix}}_{K^{\top}}
    \underbrace{
   \begin{bmatrix}
      A^1 & 0\\ 0 & A^2
   \end{bmatrix}}_{\bar{A}}
   \underbrace{ 
    \begin{bmatrix}
       1 & 0 & 0\\ 
       0 & 1 & 0\\
       0 & 1 & 0\\
       0 & 0 & 1\\
   \end{bmatrix}}_{K}
   \underbrace{ 
    \begin{bmatrix}
       {\bar{x}}_{1}(t) \\ 
       {\bar{x}}_{2}(t)\\
       {\bar{x}}_{3}(t)
   \end{bmatrix}}_{\Bar{x}(t)}.
\end{align}
Note that in the composite system, the composition matrix $K$ duplicates the shared state $\Bar{x}_2$ twice, and maps the copies to the 
dynamics of the two subsystems separately. Meanwhile, $K$ maps the non-shared state $\Bar{x}_1(t)$ to $x^1_1$, 
and it maps the non-shared state $\Bar{x}_3(t)$ to $x^2_2$.
After updating the states through $A_1$ and $A_2$, we compose the changes of all states through $K^{\top}$. The matrix $K^\top$ maps the dynamics of the non-shared variable $\dot{x}^1_1$ 
to $\dot{\bar{x}}_1$. It also maps the dynamics of the non-shared variable
$\dot{x}^2_2$ to $\dot{\bar{x}}_3$. For dynamics of the shared variable $\dot{\Bar{x}}_2(t)$,  the matrix $K^{\top}$ sums the dynamics of the variables from the subsystems, i.e., $\dot{\Bar{x}}_2(t) = \dot{\Bar{x}}^1_2(t)+\dot{\Bar{x}}^2_1(t)$.
\end{example}

Non-shared states have the same dynamics as the corresponding states in the subsystems. Shared states, on the other hand, are duplicated and summed, leading to identical values
for them within the subsystems that share them. Further, the dynamics of a shared state are the summation of the dynamics of its corresponding states in the subsystems.
%\tyler{Potential TODO: move some of the blurbs about how comopsing means duplicating the shared variable, sending to subsystems separetly, then summing the individual contribution, up here.}

\section{Compositionality of Linear Control Systems}
\label{section3}
In this section, we propose a way of composing control inputs based on the framework of the resource sharing machines.
In order to solve Problems~\ref{prob:1} and ~\ref{prob:2}, we will propose a mechanism to compose the control inputs of closed-loop subsystems based on Definition~\ref{Def:op-contr}. Then, we will investigate the conditions under which the resulting composite control system preserves the compositionality of the feedback laws of the subsystems.
We note that these are generic feedback laws and we do not consider LQR explicitly until the next section. 
We consider the case where the control signals can be designed through a full state feedback mechanism, where
\begin{align}
\label{Eq: closed-sub}
    u^1(t) = F^1 x^1(t) \quad\textnormal{ and }\quad  u^2(t) = F^2 x^2(t), 
\end{align}
where the matrices $F^i \in \mathbb{R}^{m_i \times n_i}$, $i\in \underline{2}$, are gain matrices. 
Building on Definition~\ref{Def:op-contr},
we introduce the following definition for the composite closed-loop control system. 
\begin{definition}[Composite Closed-loop Feedback Control Systems]
\label{Def:closed-contr}
For two closed-loop feedback control subsystems $\dot{x}^i (t) = {A}^i x^i(t) + {B}^iF^ix^i(t)$, $i \in \underline{2}$, the composite closed-loop system is given by 
\begin{align}
\label{eq:comp_closed}
 \dot{\bar{x}}(t)=\mathcal{A}\bar{x}(t)+\mathcal{B}
\underbrace{
\underbrace{
  \begin{bmatrix}
     F_1 & \boldsymbol{0} \\ \boldsymbol{0} & F_2
 \end{bmatrix}}_{\bar{F}} K}_{\mathcal{F}} \bar{x}(t).
 \end{align}
\end{definition}

Our choice to compose closed-loop systems in this way is based on 
wanting to preserve the dynamics of the individual subsystems.
The following result shows that this choice of composition
does indeed do so.

\begin{lemma}
\label{lem:comp_closed}
The dynamics of the composed closed-loop control systems 
$\dot{\bar{x}}(t) = \mathcal{A} \bar{x}(t) +  \mathcal{B}  \mathcal{F}\bar{x}(t)$ in \eqref{eq:comp_closed}, are equivalent to the composition of the subsystems
$\dot{x}^i (t) = {A}^i x^i(t) + {B}^iF^ix^i(t)$, $i \in \underline{2}$.
\end{lemma}
\begin{proof}
    Based on Definition~\ref{Def:op-contr} and \eqref{eq:comp_closed}, we have the composed dynamics given by 
\begin{align}
 \dot{\bar{x}}(t)  &= K^{\top}
\begin{bmatrix}
A^1 & \boldsymbol{0} \\
\boldsymbol{0} & A^2 \\
\end{bmatrix} K \bar{x}(t) + K^{\top}
\begin{bmatrix}
B^1 & \boldsymbol{0} \\
\boldsymbol{0} & B^2 \\
\end{bmatrix}
\bar{F}K\Bar{x}(t),\nonumber
\end{align}
which leads to
\begin{align}
 \dot{\bar{x}}(t)  &= 
 K^{\top} \left(
\begin{bmatrix}
A^1 & \boldsymbol{0} \\
\boldsymbol{0} & A^2 \\
\end{bmatrix} \begin{bmatrix}
x^1(t) \\
x^2(t)  \\
\end{bmatrix} +
\begin{bmatrix}
B^1F^1 & \boldsymbol{0} \\
\boldsymbol{0} & B^2F^2 \\
\end{bmatrix}
\begin{bmatrix}
x^1(t) \\
x^2(t) \\
\end{bmatrix}\right)\nonumber\\ 
&= K^{\top}\begin{bmatrix}
\dot{x}^1(t) \\
\dot{x}^2(t) \\
\end{bmatrix},
\end{align}
where $\dot{x}^1(t)$ and $\dot{x}^2(t)$ correspond to the closed-loop subsystem dynamics in Definition~\ref{Def:closed-contr}, respectively. Note that we used the fact that $K\bar{x}(t) = [x^1(t) \ x^2(t)]^{\top}$ and the definition of $\bar{F}$ in~\eqref{eq:comp_closed}. Hence, we have shown that by composing the dynamics of the two closed-loop control subsystems $\dot{x}^i (t) = {A}^i x^i(t) + {B}^iF^ix^i(t)$, $i \in \underline{2}$, %\eqref{Eq:Sys_1} and \eqref{Eq:Sys_2}, 
the composed dynamics is equivalent to the composed system  in~\eqref{eq:comp_closed}. 
\end{proof}
% $\dot{\bar{x}}(t)=\mathcal{A}\bar{x}(t)+\mathcal{B}u(t)$
% in \eqref{Eq:Contr_comp}.
    % We can leverage the same proof procedure in Proposition~\ref{lem:com_open_contr} to show this lemma. Thus, we omit the proof here.
    % \mh{If we delete the stuff in red, should we add a proof here? It seems relatively simple to prove this. Or is it true by definition?}
    % \mh{Also, it seems that this result solves Problem 1. If so, let's call it a ``Proposition''.}

%\begin{remark}
%Lemma~\ref{lem:comp_closed} shows that the feedback laws of the closed-loop subsystems are preserved through the formulation of resource sharing machines.
%Thus, resource sharing machines provide a method for formulating a composite control system from its subsystems.
%   % shows that the control design of the closed-loop subsystems will be preserved through the formulation of resource sharing machines. Therefore, it provide a way of designing a composed control system through its subsystems.  
%\end{remark}
% add stability analysis here.
% \mh{At this point, have we solved Problem~1? If so, let's say that.}

Lemma~\ref{lem:comp_closed} shows that we can leverage the composite closed-loop control system given in~\eqref{eq:comp_closed} to study the dynamics of its subsystems, and vice versa. Hence, Lemma~\ref{lem:comp_closed} solves Problem~1.

Using this formulation, we next solve Problem~\ref{prob:2}. 
\begin{theorem}
\label{Thm:Stability}
Consider two asymptotically stable closed-loop feedback control systems $\dot{x}^i (t) = {A}^i x^i(t) + {B}^iF^ix^i(t)$ for $i \in \{1,2\}$.
If the matrices ${A}^i + {B}^iF^i$ for $i \in \{1,2\}$ are symmetric, then
the composite closed-loop system 
 $\dot{\bar{x}}(t)=\mathcal{A}\bar{x}(t)+\mathcal{B}\mathcal{F}\Bar{x}(t)$ 
 with~$\mathcal{A} = K^{\top}\bar{A}K$ and~$\mathcal{BF} = K^{\top}\bar{B}\bar{F}K$
 is asymptotically stable, where~$\bar{F}$ and~$\mathcal{F}$ are from Definition~\ref{Def:closed-contr}
 and~$\bar{A}$ and~$\bar{B}$ are from Definition~\ref{Def:op-contr}. 
 \end{theorem}
\begin{proof}
    If both subsystems~\eqref{Eq:Sys_1} and~\eqref{Eq:Sys_2} with the feedback control laws given by~\eqref{Eq: closed-sub} are stable, then the system matrices
     $A^i + B^iF^i$, $i\in \{1,2\}$ are Hurwitz matrices, i.e., all eigenvalues of both matrices have strictly negative real parts. 
     Thus, based on~\eqref{Eq:Contr_comp} and~\eqref{eq:comp_closed}, 
     %$\mathcal{A}+\mathcal{B}\mathcal{F}$ is also a Hurwitz matrix. Therefore, the spectra of the matrix $\mathcal{A}+\mathcal{B}\mathcal{F}$ are located on the left-hand side of the complex plane.  
    %Based on~\eqref{Eq:Contr_comp} and~\eqref{Def:closed-contr}, 
   the block diagonal matrix $\Bar{A} + \Bar{B}\bar{F}$ is also a Hurwitz matrix
   because its eigenvalues are those of~$A^1 + B^1F^1$ and~$A^2 + B^2F^2$. 
   Further, if ${A}^i + {B}^iF^i$ for $i \in \{1,2\}$ are both symmetric, 
   then the diagonal matrix
   $\Bar{A} + \Bar{B}\bar{F}$ is also symmetric, and, because
   it is Hurwitz, it is thus negative definite as well.
   Hence, for any nonzero vector $z\in \mathbb{R}^{n_1+n_2}$,
    \begin{align}
    \label{Eq:Neg_Def}
       z^\top (\Bar{A} + \Bar{B}\bar{F})z <0.
    \end{align}
Based on the definition of the composition matrix $K$, for any nonzero vector $\Bar{z}\in \mathbb{R}^{n_1+n_2-k}$, we always have a corresponding nonzero vector $z= K\Bar{z}\in \mathbb{R}^{n_1+n_2}$. Hence, for any nonzero vector $\Bar{z}\in \mathbb{R}^{n_1+n_2-k}$, based on \eqref{Eq:Neg_Def}, we have 
\begin{align*}
       {\Bar{z}}^\top K^\top (\Bar{A} + \Bar{B}\bar{F})K\Bar{z}
        = z^\top (\Bar{A} + \Bar{B}\bar{F})z<0.
\end{align*}
Further, for any nonzero vector $\Bar{z}\in \mathbb{R}^{n_1+n_2-k}$, we have
\begin{align*}
       {\Bar{z}}^\top K^\top (\Bar{A} + \Bar{B}\bar{F})K\Bar{z}
        &= {\Bar{z}}^\top  (K^\top\Bar{A}K + K^\top\Bar{B}\bar{F}K)\Bar{z}\\
       & = \Bar{z}^\top(\mathcal{A} +\mathcal{B}\mathcal{F})\Bar{z} <0.
\end{align*}
% we can always have a corresponding nonzero vector $z= K\Bar{z}\in \mathbb{R}^{n_1+n_2}$.

%     \begin{align*}
%        {\Bar{z}}^\top K^\top (\Bar{A} + \Bar{B}\bar{F})K\Bar{z}
%         &= {\Bar{z}}^\top  (K^\top\Bar{A}K + K^\top\Bar{B}\bar{F}K)\Bar{z}\\
%        & = \Bar{z}^\top(\mathcal{A} +\mathcal{B}\mathcal{F})\Bar{z} <0.
%     \end{align*}
Because $\Bar{A} + \Bar{B}\bar{F}$ is symmetric, we see that the matrix $\mathcal{A} +\mathcal{B}\mathcal{F}$ is also symmetric. Hence, $\mathcal{A} +\mathcal{B}\mathcal{F}$ is a negative definite matrix and a Hurwitz matrix. Recall from~\eqref{eq:comp_closed} that $\mathcal{A} +\mathcal{B} \mathcal{F}$ 
defines the dynamics of the composite closed-loop control system. Thus, the composite closed-loop control system is asymptotically stable, as desired. 
\end{proof}

%Definitions~\ref{Def:op-contr} and \ref{Def:closed-contr}, along with Proposition~\ref{lem:com_open_contr} and Lemma~\ref{lem:comp_closed} solved 
%Problem~1 and partially solved Problem~2. In particular, 
%Lemma~\ref{lem:comp_closed} illustrates that the compositionality of the dynamics of 
%closed-loop feedback control systems can be preserved under resource sharing machines. In addition, 
Theorem~\ref{Thm:Stability} shows that the stability of the closed-loop 
control systems  under the composition of resource sharing machines will be preserved if the system
matrices of the closed-loop control systems are symmetric.
However, Theorem~\ref{Thm:Stability} does not determine if feedback control laws designed directly for the composite
system are the same as the compositions of control laws designed for each subsystem. 
We will answer this question in the following theorem. 

\begin{theorem}
\label{thm:comp}
Consider the composite feedback control system  $\dot{\bar{x}}(t) = \mathcal{A} \bar{x}(t) +  \mathcal{B}\bar{u}(t)$, where $\bar{u}(t)$ is designed for the composite system directly so that $\bar{u}(t) = \hat{\mathcal{F}}\Bar{x}(t)$. Consider another composite feedback control system $\dot{\bar{x}}(t) = \mathcal{A} \bar{x}(t) +  \mathcal{B}\bar{F}K\bar{x}(t)$, where the control design of the composite system is obtained through the composition of the two closed-loop control subsystems via Definition~\ref{Def:closed-contr}. 
The two systems are equivalent if $\hat{\mathcal{F}}= \bar{F}K$.
\end{theorem}
\begin{proof}
    First we consider control design for the composite system  $\dot{\bar{x}}(t) = \mathcal{A} \bar{x}(t) +  \mathcal{B}\bar{u}(t)$ directly. Based on the composite open-loop system $\dot{\bar{x}}(t) = \mathcal{A} \bar{x}(t) +  \mathcal{B}\bar{u}(t)$, we denote the feedback control law obtained for the composite system by $\bar{u}(t) = \hat{\mathcal{F}}\bar{x}(t)$. For two subsystems $\dot{x}^i (t) = {A}^i x^i(t) + {B}^iu^i(t)$ with $i \in \{1,2\}$, with the control laws given by $u^i(t) = F^ix^i(t)$, the composite control system is given by \eqref{eq:comp_closed}. Hence, 
    composing the control laws from the subsystems gives $\mathcal{F}\bar{x}(t)= \bar{F}K\bar{x}(t)$. Comparing the dynamics of the two systems, 
    if $\bar{F}K = \hat{\mathcal{F}}$, then the systems have identical dynamics. 
\end{proof}
Theorem~\ref{thm:comp} provides a way of testing the compositionality of a feedback control design, namely by comparing
the two feedback control gain matrices $\hat{\mathcal{F}}$ and $\mathcal{F}=\Bar{F}K$. The theorem lays a foundation for 
designing controllers for composed systems via resource sharing machines through designing controllers for their subsystems. 
% Additionally, if $\hat{\mathcal{F}}$ and $\mathcal{F}$ are different, we could potentially leverage $||\hat{\mathcal{F}} - \mathcal{F}||$ to characterize the gap between the two control design method. 
%Together, Theorems~\ref{Thm:Stability} and~\ref{thm:comp} solve Problem~2. 
Based on the results in this section, we will 
use resource sharing machines to 
study the compositionality of LQR designs specifically in the next section.

\section{Compositionality of LQR} 
In this section, we study the compositionality of LQR design problems through resource sharing machines. We will analyze 
differences between (i) the LQR controller designed for a composite system and the 
(ii) controller that results from composing LQR controllers
designed for the subsystems of the composite. 
Together, these analyses will solve Problem~\ref{prob:3}. 

\subsection{LQR Design for Subsystems} \label{ss:lqr_subsystems}
First we introduce the LQR problem for each subsystem. Again, we consider two linear subsystems $\dot{x}^i (t) = {A}^i x^i(t) + {B}^iu^i(t)$ for
$i \in \{1,2\}$, and we define cost functions
%\mh{Just integrate from~$0$ to~$\infty$.}
\begin{multline}
\label{eq:cost_fun_sub}
     J^i(x^i , u^i) = \int_{0}^{\infty} x^i(\tau)^{\top} Q^i x^i(\tau)\,d\tau \\ 
    + \int_{0}^{\infty} u^i(\tau)^{\top} R^i u^i(\tau)\,d\tau,  
\end{multline}
where~$Q^i\in\mathbb{R}^{n_i\times n_i}$ and $R^i\in\mathbb{R}^{m_i\times m_i}$ for
$i\in \{1,2\}$. 
%\mh{Do we need each matrix to have non-negative entires? Also, make sure~$\mathbb{R}_{\geq 0}$ is defined in the notation section.}
Additionally, the matrix $Q^i$ is positive-semidefinite while $R^i$ is positive definite for $i\in \{1,2\}$. 
%The first term on the right-hand side of \eqref{eq:cost_fun_sub} is the terminal state cost. 
The first term on the right-hand side is the penalty on excessive state size, and the second term on the right-hand side 
is the penalty on control effort. 
%The foundations of LQR design present the explicit solution to the optimal control problem that minimizes the cost functions in \eqref{eq:cost_fun_sub}. 
%\baike{We consider the infinite-horizon case in this work, i.e., $t_f = +\infty$.}
The optimal solutions are given by the following state-feedback representations \cite{anderson2007optimal}:
\begin{equation}
\label{Eq: LQR_sub}
    u^i(t) = -(R^i)^{-1}(B^i)^{\top}P^ix^i(t)= -F^ix^i(t),
\end{equation}
for $i\in\{1,2\}$, where~$P^i$ for $i\in\{1,2\}$ satisfies an algebraic Riccati equation given by 
% \begin{align}
%      \dot{P}^i(t)&= -P^i(t) A^i -(A^i)^{\top}P^i(t)-Q^i\\ \nonumber
%     &+ P^i(t)B^i (R^i)^{-1}(B^i)^{\top}P^i(t). 
% \end{align}
% In order to simplify the analysis, we use
% constant gains calculated using the steady-state solution 
% to the algebraic Riccati equation, given by
\begin{equation}
     \boldsymbol{0} = -P^i A^i -(A^i)^{\top}P^i-Q^i+ P^iB^i (R^i)^{-1}(B^i)^{\top}P^i, 
     \label{Eq:ARE_sub}
\end{equation}
%\mh{Why not just set~$t_f = \infty$? Then the optimal~$P$ is the constant one. Would that work here?}
where the solution $P^i$ is positive semidefinite. 

\subsection{LQR Design for the Composite System} \label{ss:lqr_composite}
We now define the LQR problem for the composite system in~\eqref{eq:comp_closed}. The cost function is defined as 
\begin{multline}
\label{eq:cost_fun}
     \mathcal{J}(\bar{x},\bar{u}) = 
     %\bar{x}(t_f)
     %\underbrace{K^{\top}
     %\begin{bmatrix}
     %    H^1 & \boldsymbol{0} \\  \boldsymbol{0} & H^2
     %\end{bmatrix} K}_{\mathcal{H}}
     %\bar{x}(t_f) \ \nonumber \\ 
     \int_{0}^{\infty} \bar{x}(\tau)^{\top} 
     \underbrace{K^{\top}
     \begin{bmatrix}
         Q^1 & \boldsymbol{0} \\  \boldsymbol{0} & Q^2
     \end{bmatrix}K}_{\mathcal{Q}}
     \bar{x}(\tau)\,d\tau  \\ 
    + \int_{0}^{\infty} \bar{u}(\tau)^{\top} 
    \underbrace{
     \begin{bmatrix}
         R^1 & \boldsymbol{0} \\  \boldsymbol{0} & R^2
     \end{bmatrix}}_{\bar{R}}
    \bar{u}(\tau)\,d\tau,  
\end{multline}
where $\mathcal{Q}$ and $\bar{R}$ are composite weight matrices. 
Note that following the proof of Theorem~\ref{Thm:Stability}, 
the matrix
$\mathcal{Q}$ is positive semidefinite. And $\bar{R}$ is 
positive definite. Further, the optimal solution of the control problem in \eqref{eq:comp_closed} with the objective function $\mathcal{J}$ is
\begin{equation}
\label{Eq: LQR_comp}
    \bar{u}(t) = -\bar{R}^{-1}\mathcal{B}^{\top}\mathcal{P}\Bar{x}(t)= -\mathcal{F}\bar{x}(t),
\end{equation}
where $\mathcal{P}$ is a symmetric positive semidefinite matrix that solves the algebraic Riccati equation %\baike{($t_f = +\infty$)}
\begin{equation}
\label{Eq:ARE_comp}
     \bold{0} = -\mathcal{P}\mathcal{A} -\mathcal{A}^{\top}\mathcal{P}-\mathcal{Q}+ \mathcal{P}\mathcal{B} \bar{R}^{-1}\mathcal{B}^{\top}\mathcal{P}. 
\end{equation}

\subsection{Comparison of LQR Approaches}
This subsection compares the LQR designs from Sections~\ref{ss:lqr_subsystems} and~\ref{ss:lqr_composite}. 
The following theorem gives a necessary and sufficient condition for the compositionality of the LQR controller,
i.e., conditions under which the composition of LQR controllers designed for subsystems is equivalent
to an LQR controller designed for the composite system. 

\begin{theorem} \label{thm:LQR_Comp}
The LQR solution of the composite system in \eqref{eq:comp_closed} and
the composition of the LQR solutions for the subsystems,
given in \eqref{Eq: LQR_sub}, are equivalent if and only if
\begin{equation} \label{eq:K_condition}
\underbrace{\begin{bmatrix}
    P_1 & \boldsymbol{0}\\\boldsymbol{0} & P_2 
\end{bmatrix}}_{\Bar{P}}K
= K\mathcal{P}.
\end{equation}
\end{theorem}
\begin{proof}
    We first introduce the controller that results from composing
    the LQR controllers designed for each subsystem. 
    By substituting \eqref{Eq: LQR_sub} into \eqref{Eq:Contr_comp} we have 
\begin{multline} \label{eq:lqr_subsystems}
 \dot{\bar{x}}(t)=\mathcal{A}\bar{x}(t) \\ +\mathcal{B}
  \begin{bmatrix}
-(R^1)^{-1}(B^1)^{\top} & \boldsymbol{0} 
\\ \boldsymbol{0} & -(R^2)^{-1}(B^2)^{\top}
 \end{bmatrix}
 \Bar{P}
 %  \begin{bmatrix}
 %    P^1 & 0 \\ 0 & P^2
 % \end{bmatrix}
 K\Bar{x}(t).
\end{multline}
Next, by implementing the LQR design for the composite system (given in \eqref{Eq: LQR_comp}) on the system in \eqref{eq:comp_closed}, we have
\begin{align} 
 \dot{\bar{x}}(t)&=\mathcal{A}\bar{x}(t)+\mathcal{B}
(-\Bar{R}^{-1}\mathcal{B}^{\top}\mathcal{P})\Bar{x}(t)\\
&=\mathcal{A}\bar{x}(t) \\
&\quad+\mathcal{B}
  \begin{bmatrix}
    -(R^1)^{-1}(B^1)^{\top} & \boldsymbol{0} \\ \boldsymbol{0} & -(R^2)^{-1}(B^2)^{\top}
 \end{bmatrix}
 K\mathcal{P}
\Bar{x}(t). \label{eq:lqr_composite}
\end{align}
Note that we use the fact that 
$\mathcal{B}=K^\top \bar{B}$. 
In order to compute the difference between the system 
behaviors under the two different LQR designs, we subtract the right-hand
side of~\eqref{eq:lqr_composite} from the right-hand side of~\eqref{eq:lqr_subsystems} to find
\begin{align}
\label{Eq:condition}
\mathcal{B}
  \begin{bmatrix}
    -(R^1)^{-1}(B^1)^{\top} & \boldsymbol{0} \\ \boldsymbol{0} & -(R^2)^{-1}(B^2)^{\top}
    \end{bmatrix}
    (
%     \begin{bmatrix}
%     P_1 & 0\\0 & P_2 
% \end{bmatrix}
\bar{P}K-K\mathcal{P}) = \boldsymbol{0}.
\end{align}
Since $B^1$ and $B^2$ are non-zero matrices and $R^1$ and $R^2$ are positive definite matrices,
\eqref{Eq:condition} is satisfied if and only if $
% \begin{bmatrix}
%     P_1 & 0\\0 & P_2 
% \end{bmatrix}K
\Bar{P}K=K\mathcal{P}$. Therefore, we have shown that the composition of the LQR solutions in \eqref{Eq: LQR_sub} for the control subsystems via the composition \eqref{eq:comp_closed} and the direct LQR design of the composite system in \eqref{eq:comp_closed} are equivalent if and only if $\Bar{P}K$ 
= $K\mathcal{P}$. 
% we need different notations for two types of composed systems in this proof.
\end{proof}
Theorem~\ref{thm:LQR_Comp} can be regarded as a special case of Theorem~\ref{thm:comp} in which the control law is based on LQR design. However, Theorem~\ref{thm:LQR_Comp}  further illustrates that, for a composite system, it is possible to design an LQR controller for the system using the LQR designs of the corresponding control subsystems, if we can construct a composition matrix $K$ to satisfy~\eqref{eq:K_condition}. Note that the matrices $P_1$ and $P_2$ that comprise $\Bar{P}$ and $\mathcal{P}$ are from the steady-state solutions of the algebraic Riccati equations given in~\eqref{Eq:ARE_sub} and~\eqref{Eq:ARE_comp}, respectively. Therefore, we further investigate the compositionality of LQR through exploring the properties of these algebraic Riccati equations.
Specifically, the next theorem unifies them into a single Generalized Algebraic Riccati equation.
\begin{theorem}
\label{Thm:G_ARE}
The matrix $K\mathcal{P}$ in the LQR design for the composite system  and the matrix $\Bar{P}K$ in the composition of LQR designs for the corresponding subsystems are
both solutions to the General Algebraic Riccati equation 
\begin{align}
\label{Eq: ARE_constru}
\boldsymbol{0} = - X^\top\Bar{A}K - K^\top\Bar{A}^\top X  -
 \mathcal{Q} + X^\top\Bar{B}\Bar{R}^{-1} \Bar{B}^\top X,
\end{align}
where the matrix~$X$ is to be solved for. 
\end{theorem}
\begin{proof}
%To further explore the conditions to guarantee $\Bar{P}K$ = $K\mathcal{P}$, 
We first combine the  algebraic Riccati equations of control subsystems in \eqref{Eq:ARE_sub} with the solution $\bar{P}$ in~\eqref{Eq:ARE_sub},  %\mh{from which equation?}, 
given by
 % $\begin{bmatrix}
 %    P^1 & 0 \\ 0 & P^2
 % \end{bmatrix}$, where
\begin{multline} \label{Eq:comp_P}
\boldsymbol{0} = -\Bar{P} \begin{bmatrix}
    A^1 & 0 \\ 0 & A^2
 \end{bmatrix} -\begin{bmatrix}
    A^1 & 0 \\ 0 & A^2
 \end{bmatrix}^{\top}
 \Bar{P}
 -\begin{bmatrix}
    Q^1 & 0 \\ 0 & Q^2
 \end{bmatrix} 
  \\ +\Bar{P}\begin{bmatrix}
    B^1 & 0 \\ 0 & B^2
 \end{bmatrix} \begin{bmatrix}
    R^1 & 0 \\ 0 & R^2
 \end{bmatrix}^{-1}\begin{bmatrix}
    B^1 & 0 \\ 0 & B^2
 \end{bmatrix}^{\top}\Bar{P}. 
\end{multline}
Next, %for $\mathcal{P}$ from~\eqref{Eq:ARE_comp}, 
we expand the algebraic Riccati equation of the LQR design for the composite system in~\eqref{Eq:ARE_comp}, given by
%\mh{Where does this next equation come from?}
\begin{align}
\boldsymbol{0} &= -\mathcal{P}
     K^{\top}
    \begin{bmatrix}
        A^1 & 0 \\ 0 & A^2
    \end{bmatrix} K
     - K^{\top}\begin{bmatrix}
        A^1 & 0 \\ 0 & A^2
    \end{bmatrix}^{\top} K\mathcal{P} \\
    &-K^{\top}
    \begin{bmatrix}
        Q^1 & 0 \\ 0 & Q^2
    \end{bmatrix} K\\ \nonumber
    &+ \mathcal{P}K^{\top}
    \begin{bmatrix}
        B^1 & 0 \\ 0 & B^2
    \end{bmatrix} 
    \begin{bmatrix}
        R^1 & 0 \\ 0 & R^2
    \end{bmatrix} ^{-1}
    \begin{bmatrix}
        B^1 & 0 \\ 0 & B^2
    \end{bmatrix}^{\top} K
    \mathcal{P}.  
\end{align}
It is not straightforward to compare the solutions of these two algebraic Riccati equations since 
$\Bar{P}\in\mathbb{R}^{(n_1+n_2)\times (n_1+n_2)}$ and $\mathcal{P}\in\mathbb{R}^{(n_1+n_2-k)\times (n_1+n_2-k)}$ have different dimensions. Additionally, we want to compare $\Bar{P}K$ and $K\mathcal{P}$. Therefore, we will construct a Generalized Algebraic Riccati equation whose solutions will include both $\Bar{P} K$ and $K\mathcal{P}$.

We left multiply $K^\top$ to the right-hand side of~\eqref{Eq:comp_P} and we right multiply $K$ to the right-hand side of~\eqref{Eq:comp_P}. Then we find that both $\bar{P}K$ and $K\mathcal{P}$ are solutions to the Generalized Algebraic Riccati equation, given by
\begin{align}
\label{Eq: ARE_constru}
\boldsymbol{0} = -X^\top\Bar{A}K - K^\top\Bar{A}^\top X -
 \mathcal{Q} +  X^\top\Bar{B}\Bar{R}^{-1} \Bar{B}^\top X,
\end{align}
where $X$ is to be solved for. Based on~\eqref{Eq: ARE_constru}, the matrix $K\mathcal{P}$ from the LQR design of the composite system 
and the matrix $\Bar{P}K$ from the composed LQR designs of the corresponding subsystems are both solutions to the Generalized Algebraic Riccati equation. 
\end{proof}

Theorem~\ref{Thm:G_ARE} connects  the LQR design for the composite system and the composed LQR designs for the corresponding subsystems. 
Note that it is common for an algebraic Riccati equation to have more than one solution~\cite{kuvcera1973review}. Hence,
we are able to consider the case in which $\bar{P}K \neq K\mathcal{P}$ in the Generalized Algebraic Riccati equation in~\eqref{Eq: ARE_constru}. %\mh{Can you add a reference for that?}. 
%Hence, 
We next further leverage \eqref{Eq: ARE_constru} to determine conditions under which $\bar{P}K = K\mathcal{P}$ holds. 
%We will explore properties of $\bar{P}K$ and $K\mathcal{P}$ to check the condition that $K^\top\bar{P} = \mathcal{P}K^\top$.

\begin{corollary}
\label{Cor:sym}
    The equation $\bar{P}K = K\mathcal{P}$ holds only if $\bar{P}KK^\top$ is symmetric and positive semidefinite.
\end{corollary}
\begin{proof}
    If the two solutions to \eqref{Eq: ARE_constru} are equivalent, i.e., $\bar{P}K = K\mathcal{P}$, we must have  $\bar{P}KK^\top = K\mathcal{P}K^\top$. Since $\mathcal{P}$ is a  symmetric, positive semidefinite matrix due to being the solution to the algebraic Riccati equation in~\eqref{Eq:ARE_comp},
     we have that $K\mathcal{P}K^\top$ is also symmetric. Based on the proof of Theorem~\ref{Thm:Stability}, if $\mathcal{P}$ is a symmetric and positive semidefinite matrix, then $K\mathcal{P}K^\top$ must be a symmetric and positive semidefinite matrix. Hence, to ensure that $\bar{P}KK^\top = K\mathcal{P}K^\top$, the matrix $\bar{P}KK^\top$ must be symmetric and positive semidefinite. 
     %Hence, $\bar{P}KK^\top$ being a symmetric positive semidefinite matrix is a necessary condition for $\bar{P}K = K\mathcal{P}$. 
\end{proof}
Corollary~\ref{Cor:sym} provides a necessary condition to determine whether a composed LQR design from control subsystems is the same as the LQR design of the composite system. In particular, based on Corollary~\ref{Cor:sym}, it is unnecessary for us to compute the LQR controller of the composite system to check compositionality. Instead, we can explore the compositionality of the LQR design through studying the algebraic Riccati equations of the LQR designs of the  subsystems in \eqref{Eq:ARE_sub} and the composition matrix $K$. 

Following this idea, we present the following theorem, where we can leverage information from the control subsystems to determine the compositionality of LQR, without
needing to actually generate any LQR designs or solve any Riccati equations. 
%we will explore the gap between the two control designs if $\bar{P}K \neq K\mathcal{P}$. We will construct a bound to characterize the gap between the two types of control designs.
\begin{theorem}
\label{Thm:Compositionality}
Let $\bar{P}KK^\top$ be a positive semidefinite, symmetric matrix. If the pair $(\Bar{A}KK^\top, \Bar{B}\Sigma\Lambda^{-\frac{1}{2}})$ is a controllable pair, and the pair $(\Bar{A}KK^\top, \Delta$) is an observable pair, where $\Sigma$ and $\Lambda$ are obtained through 
%\mh{Where is~$\bar{R}$ used here?} 
the eigendecomposition of the symmetric positive definite matrix $\bar{R}=\Sigma\Lambda\Sigma^\top$ in~\eqref{eq:cost_fun}, 
and $K\mathcal{Q}K^\top = \Delta^\top \Delta$,  
then the diagram in Definition~\ref{Def:compositionality} commutes.
That is, the following two objects are equivalent: 
(i) the composition of the LQR solutions \eqref{Eq: LQR_sub} of the control subsystems via \eqref{eq:comp_closed} and (ii) the LQR solution of the composite system in \eqref{eq:comp_closed}.
\end{theorem}
\begin{proof}
In order to compare the solutions, we will explore the solutions of the Generalized Algebraic Riccati equation~\eqref{Eq: ARE_constru} 
proposed in Theorem~\ref{Thm:G_ARE}. Unlike standard algebraic Riccati equations, where the  matrices are always square, the Generalized Algebraic Riccati equation constructed in~\eqref{Eq: ARE_constru} has a non-square generalized system matrix $\Bar{A}K$ and the solution $X$ is also a nonsquare matrix. Hence, existing techniques to characterize the standard Algebraic Riccati equation cannot be applied directly to study~\eqref{Eq: ARE_constru}. Hence, instead of studying solutions to the Generalized Algebraic Riccati equation in~\eqref{Eq: ARE_constru}, we will study the following constructed Algebraic Riccati equation:
\begin{align}
\label{eq:ARE_Final}
\boldsymbol{0} &= -\underbrace{KX^\top}_{\mathcal{X}}\Bar{A}KK^\top - KK^\top\Bar{A}^\top \underbrace{X K^\top}_{\mathcal{X}} -
 K\mathcal{Q}K^\top\\ \nonumber
 &+ \underbrace{KX^\top}_{\mathcal{X}} \Bar{B}\Bar{R}^{-1} \Bar{B}^\top \underbrace{X K^\top}_{\mathcal{X}},
\end{align}
%\mh{What are the underbraces for?}
where $\mathcal{X}=\Bar{P}KK^\top$ or $\mathcal{X}=K\mathcal{P}K^\top$ is a solution to the equation. 
Note that~\eqref{eq:ARE_Final} is obtained by left multiplying and right multiplying the Generalized Algebraic Riccati equation in~\eqref{Eq: ARE_constru} 
by $K$ and $K^{\top}$, respectively.

Now we explore conditions under which the two solutions $\Bar{P}KK^\top$ and $K\mathcal{P}K^\top$  will be the same. We consider \eqref{eq:ARE_Final} as an Algebraic Riccati equation of a linear system with system matrix $\Bar{A}KK^\top$ and input matrix $\Bar{B}\Bar{R}^{-1} \Bar{B}^\top$. 
Further, we consider a standard LQR formulation for 
\eqref{eq:ARE_Final} %has the corresponding LQR formulation  that is similar to 
as in \eqref{eq:cost_fun_sub}, where 
%\mh{What is the preceding sentence saying? Are we just saying that we consider a standard LQR cost? If so, just say that.}
%In the similar LQR formulation,
we have the weight matrix for the inputs as $\Bar{R}$ and the weight matrix for the states as $KQK^\top$.
By comparing~\eqref{Eq:ARE_comp} and~\eqref{eq:ARE_Final}, if we substitute 
$\mathcal{A}$ with $\bar{A}KK^\top$, $\mathcal{P}$ with $\mathcal{X}$, $\mathcal{B}$ with $\bar{B}$, and $\mathcal{Q}$ with $KQK^\top$,
%Further,
the solution to the Algebraic Riccati equation in \eqref{eq:ARE_Final} is given by $\mathcal{X} = XK^\top$.
Note that based on Theorem~\ref{Thm:G_ARE}, $X$ can be $\Bar{P}K$ or $\mathcal{P}K^\top$.
Hence, $\mathcal{X} = \Bar{P}KK^\top$ and $\mathcal{X} = K\mathcal{P}K^\top$ can both be solutions to~\eqref{eq:ARE_Final}. In addition, based on Corollary~\ref{Cor:sym}, 
$K\mathcal{P}K^\top$ is a positive semidefinite matrix, and
the necessary condition for $\mathcal{X} = \Bar{P}KK^\top= K\mathcal{P}K^\top$ is that $\Bar{P}KK^\top$ is a positive semidefinite matrix. Thus, the prerequisite to study whether $\mathcal{X} = \Bar{P}KK^\top= K\mathcal{P}K^\top$
%\mh{It is not obvious that the solution is~$XK^{\top}$. Can you say more about why that's true?}
%We assume 
is that $\Bar{P}KK^\top$ and 
$K\mathcal{P}K^\top$ are both symmetric positive semidefinite solutions to \eqref{eq:ARE_Final}.
%\mh{Is this assumption okay to make? Also this should come sooner in the proof.}

To guarantee the existence of the solution $\mathcal{X}$ and that it is the unique symmetric positive semidefinite
solution of \eqref{eq:ARE_Final}, we need to leverage the following results for continuous algebraic Riccati equations. 
A continuous algebraic Riccati equation is
\begin{align}
    YC + C^\top Y -YDD^\top Y = -G,
\end{align}
where $C\in\mathbb{R}^{n\times n}$, $D\in\mathbb{R}^{n\times m}$, and $G\in\mathbb{R}^{n\times n}$ is a symmetric positive
semi-definite non-zero matrix. The matrix $Y\in \mathbb{R}^{m\times m} $ is what is to be solved for. 
Further, if $(C,D)$ is a controllable pair and $(C,Q)$ is an observable pair (where $G=Q^\top Q$), then $Y$ is the unique symmetric positive semi-definite
solution \cite{davies2008new}. 
Further, based on the proof of Theorem~\ref{Thm:Stability}, we can show that $K\mathcal{P}K^\top$  is a positive semidefinite matrix. 
Note that we use the fact that $\bar{R}=\Sigma\Lambda\Sigma^\top$ is the eigendecomposition of the symmetric positive definite matrix $\bar{R}$, where $\Sigma\in\mathbb{R}^{(m_1+m_2)\times(m_1+m_2)}$ is the eigenvector matrix corresponding to the diagonal eigenvalue matrix $\Lambda \in\mathbb{R}^{(m_1+m_2)\times(m_1+m_2)}$.
Therefore, under the condition 
that $KK^\top\Bar{P}$ and $K\mathcal{P}K^\top$ are both symmetric positive semidefinite
solutions to \eqref{eq:ARE_Final}, if the pair $(\Bar{A}KK^\top, \Bar{B}\Sigma\Lambda^{-\frac{1}{2}})$ is a controllable pair, and the pair $(\Bar{A}KK^\top, \Delta)$ is an observable pair, then \eqref{eq:ARE_Final} has a unique symmetric positive semidefinite solution. By uniqueness, we must have $\Bar{P}KK^\top = K\mathcal{P}K^\top$. 
%\mh{State this last sentence before~$\Sigma$ and~$\Lambda$ are used.}
%Therefore, \eqref{eq:ARE_Final} has the unique symmetric positive semi-definite solution will lead to $KK^\top\Bar{P} =K\mathcal{P}K^\top$. 

Further, based on the fact that $K$ is not a zero matrix, we have  $\Bar{P}K =K\mathcal{P}$. Through Theorem~\ref{thm:LQR_Comp} and Theorem~\ref{Thm:G_ARE}, we can conclude that the composition of the LQR solutions \eqref{Eq:Contr_comp} of the subsystems via \eqref{Eq:Contr_comp} and the LQR solution of the composite system in \eqref{Eq: LQR_comp} are equivalent.  The composition of the LQR solutions \eqref{Eq: LQR_sub} of the control subsystems via \eqref{eq:comp_closed} and the LQR solution of the composite system in \eqref{eq:comp_closed} are equivalent. 
\end{proof} 
\begin{remark}
Theorem~\ref{Thm:Compositionality} transforms the study of compositionality of LQR design to the investigation of the controllability of the constructed pair $(\Bar{A}KK^\top, \Bar{B}\Sigma\Lambda^{-\frac{1}{2}})$ and the investigation of the observability of the pair $(\Bar{A}KK^\top, \Delta)$. Note that the controllability and observability analyses of the constructed pairs are not directly performed on the composed linear systems in \eqref{eq:comp_closed}. The advantage of leveraging Theorem~\ref{Thm:Compositionality} is that, instead of comparing control designs that involve solving potentially many algebraic Riccati equations, we can directly leverage information from the system and input matrices of the subsystems, the optimal control problem formulation of the subsystems, and the composition structure, to verify the compositionality of LQR design. %Hence, for two subsystems and the corresponding LQR design problems, the way of composition, captured by the composition matrix $K$, will determine controllability of the pair $(\Bar{A}KK^\top, \Bar{B}\Sigma\Lambda^{-\frac{1}{2}})$, and the the observability of the pair $(\Bar{A}KK^\top, -KQK^\top)$. 
\end{remark}
% Theorem~\ref{Thm:Compositionality} provides a way of studying compositionality of LQR design on resource-sharing machines.  
\section{Conclusion}
In this work, we introduce the composition of linear systems using the formulation of resource sharing machines. We explore the compositionality of open-loop and closed-loop control systems, specifically focusing on the LQR design problem on linear systems. 
%The compositionality problem is transformed into the controllability and observability problem, with the method of composition playing a vital role. 
In future work, we aim to investigate how other system properties, such as stability and controllability, are affected by the composition of resource sharing machines. We are motivated to study the optimality gap of the LQR designs in the cases in which compositionality fails. 
Furthermore, we are interested in implementing resource-sharing machines in engineering applications.
% In this work, we introduce how to compose linear systems under the formulation of resource sharing machines. We study the compositionality of open-loop control systems, In particular, we explore the  compositionality of the LQR design problem on linear systems. We transform the compositionality problem into the controllability and observability problem, where the way of composition plays an important role. Future work will explore how another system properties such as stability and controllability change under the composition of resource sharing machines. Additionally, we are interested in implementing the way of resource-sharing machines in engineering applications.
\normalem
\bibliographystyle{IEEEtran}
%{\footnotesize
\bibliography{IEEEabrv,main}

%\bibliographystyle{IEEEtran}
%\bibliography{IEEEabrv,main}

\end{document}